\documentclass[12pt]{article}
\usepackage{amsmath}
\usepackage{amssymb}
\usepackage{amsthm,url,tikz}
\usepackage[utf8]{inputenc}
\usepackage[english]{babel}
\usepackage{graphicx}
\usepackage{hyperref}

\usepackage{xcolor}
\newtheorem{theorem}{Theorem}[section]
\newtheorem{proposition}{Proposition}[section]
\newtheorem{corollary}{Corollary}[section]
\newtheorem{lemma}[theorem]{Lemma}
\theoremstyle{definition}
\newtheorem{defn}{Definition}[section]

\newtheorem{remark}{Remark}[section]

\usepackage[nottoc,notlot,notlof]{tocbibind} 
\theoremstyle{remark}

\title{Galois LCD  Codes  Over $\mathbb{F}_q+u\mathbb{F}_q+v\mathbb{F}_q+uv\mathbb{F}_q$ }
\author{Astha Agrawal, Gyanendra K. Verma  and R. K. Sharma}
\date{}

\begin{document}
	\maketitle	
%
	
	\begin{abstract}
		In \cite{anote}, Wu and Shi studied $ l $-Galois LCD codes over finite chain ring  $\mathcal{R}=\mathbb{F}_q+u\mathbb{F}_q$, where $u^2=0$ and  $ q=p^e$ for some prime $p$ and positive integer $e$. In this work, we extend the results to the finite non chain ring $ \mathcal{R} =\mathbb{F}_q+u\mathbb{F}_q+v\mathbb{F}_q+uv\mathbb{F}_q$, where $u^2=u,v^2=v $ and $ uv=vu $. We define a correspondence between   $ l $-Galois dual of linear codes over $ \mathcal{R} $  and $ l $-Galois dual of its component codes over $ \mathbb{F}_q .$ Further, we construct Euclidean LCD and $ l $-Galois LCD codes from linear code over $ \mathcal{R} $. This consequently leads us to prove that any linear code over $ \mathcal{R} $ is equivalent to Euclidean ($ q>3 $) and $ l $-Galois LCD ($0<l<e$, and $p^{e-l}+1\mid p^e-1$) code over $ \mathcal{R} .$ Finally, we investigate MDS codes over $ \mathcal{R} .$	
	\end{abstract}
	\section{Introduction}
	Let $\mathcal{C}$ be a linear code which intersects its dual trivially then $\mathcal{C}$ is said to be an LCD code (shortened form for Linear Complimentary Dual code). In 1992, LCD codes were  defined and characterized by Massey   \cite{massey} over finite fields. For a two-user binary adder channel, an optimal linear coding solution is obtained by LCD codes. After this, LCD codes  gained the interest of many researchers due to its wide applications in many areas like consumer electronics, data storage communications systems, and cryptography. In 1994, Yang and Massey \cite{conditionforcycliccodemassey} derived LCD cyclic codes. Carlet and Guilley \cite{sideattackcarlet} constructed several LCD codes and presented an implementation of binary LCD codes against fault injection and side channel attacks.
	
	We refer to the vast literature   on  constructions and characterisation  of Euclidean and Hermitian LCD codes in \cite{constructionmds,newmds,cafinite,newchpr} and the references therein. Fan and Zhang \cite{kgaloisintroduce} generalized  Euclidean and Hermitian inner products to the  $ l $-Galois inner product over  finite fields. They  studied  self dual constacyclic codes for the $ l $-Galois inner product over  finite fields.	Liu et al. \cite{galoislcd} obtained  $ l $-Galois LCD codes over  finite fields, where they characterized $ \lambda $-constacyclic codes as $ l $-Galois LCD codes. In \cite{lcdovercommring},  some criteria for a linear code to be  an LCD code over  a finite commutative ring were obtained.  
	
	A linear  code $\mathcal{C}$ with parameter $ [n,k,d] $ over a finite field is said to be maximum distance separable (MDS) code if the minimum distance $ d $ of  code $\mathcal{C}$ attains Singleton bound, i.e., $ d=n-k+1 $. 
	MDS codes have very good theoretical and practical properties. Carlet et al. \cite{euclideanhermitian} discussed the existence of Euclidean LCD MDS codes over the finite field and gave several constructions of Euclidean and Hermitian LCD MDS codes.  In \cite{someresultsz4},  the authors  studied  linear codes over the ring $ \mathbb{Z}_4+u\mathbb{Z}_4+v\mathbb{Z}_4+uv\mathbb{Z}_4 $  for Euclidean inner product and   discussed some properties of Euclidean dual and  MDS codes. Several authors investigated skew cyclic codes, constacyclic codes and quantum error correcting codes over ring $ \mathcal{R} $ \cite{skewc,ConstacyclicCO,quantumcode}. 	The study of $ l $-Galois LCD codes over  finite chain ring $ \mathbb{F}_q+u\mathbb{F}_q $ was done in \cite{anote}. Authors proved that for any linear code over $ \mathbb{F}_q+u\mathbb{F}_q $ there exists equivalent Euclidean and $ l $-Galois LCD code. However, $ l $-Galois LCD codes over the finite non chain ring, $ \mathcal{R}=\mathbb{F}_q+u\mathbb{F}_q+v\mathbb{F}_q+uv\mathbb{F}_q, $ have not been studied yet. The main idea of this article is to deal with this problem. Taking inspiration from \cite{anote}, we characterize $ l $-Galois linear codes over finite non chain ring $\mathcal{R}=\mathbb{F}_q+u\mathbb{F}_q+v\mathbb{F}_q+uv\mathbb{F}_q .$ The outline of this article is arranged as follows.
	
	Section \ref{pre} contains the basic mathematical background required thereafter. We have defined an inner product which is a generalisation of Euclidean and Hermitian inner products over $\mathcal{R}$. In Section \ref{coder}, we have obtained the  $ l $-Galois LCD code over $ \mathcal{R} .$ We have also discussed basic results on a gray image of $l$- Galois LCD code and its dual.  In Section \ref{construction}, we  construct  Euclidean and $ l $-Galois LCD codes  from linear codes over $ \mathcal{R} .$  Moreover, we demonstrate that  a linear code  over $ \mathcal{R} $ is equivalent to Euclidean and $ l $-Galois LCD code over $ \mathcal{R} .$ In Section \ref{mdscode}, we have looked into the MDS codes over $ \mathcal{R} .$ We also have given results connecting $\mathcal{C}^{\perp_l}$ and $\mathcal{C}$ whenever one of them is MDS code. Finally, Section \ref{con} concludes the article with some establishments of future work that can be done using this work. 
	\section{Preliminaries} \label{pre}
	
	In this paper, we denote $ q $ as a prime power that is $q=p^e$, for some  integer $ e>0 $ and $ \mathbb{F}_q $ as a finite field of order $ q $.	Let us consider the ring $\mathcal{R}=\mathbb{F}_q+u\mathbb{F}_q+v\mathbb{F}_q+uv\mathbb{F}_q=\{a+ub+vc+uvd \mid u^2=u, v^2=v,  uv=vu, \text{ and } a,b,c,d \in \mathbb{F}_q \}$.
	It is easy to observe that $\mathcal{R}$ is a commutative and principal ideal ring.  Since it has four maximal ideals so it is  a semi-local ring and  a finite  non-chain ring. Let $\gamma_1=1-u-v+uv,\ \gamma_2=uv,\ \gamma_3=u-uv,\ \gamma_4=v-uv$ then  $\sum_{i=1}^4 \gamma_i=1$, $\gamma_i^2=\gamma_i$ and  $\gamma_i\gamma_j=0$ for $i\neq j$. By Chinese remainder theorem,  we can write $\mathcal{R}=\gamma_1\mathcal{R}\oplus\gamma_2\mathcal{R}\oplus\gamma_3\mathcal{R}\oplus\gamma_4\mathcal{R}$ and  $\gamma_i\mathcal{R}\cong\gamma_i\mathbb{F}_q\ $ for  $  i=1,2,3,4$. For any $ a\in \mathcal{R} $, $ a $ can be uniquely written as $ a=\sum_{i=1}^{4}\gamma_i a_i $ where $ a_i\in \mathbb{F}_q $ for $ i=1,2,3,4 .$ Hence $\mathcal{R}\cong\gamma_1\mathbb{F}_q\oplus\gamma_2\mathbb{F}_q\oplus\gamma_3\mathbb{F}_q\oplus\gamma_4\mathbb{F}_q$.
	\begin{defn}
		A code over $\mathcal{R}$ of length $n$ is a non empty subset of $\mathcal{R}^n$. The code $\mathcal{C}$ is said to be linear if it is $\mathcal{R}$-submodule of $\mathcal{R}^n$.
	\end{defn}
	
	\begin{defn}
		The Hamming weight of $x=(x_1,x_2,\dots,x_n) \in \mathbb{F}_q^n$ is defined to be the cardinality of non-zero $x_i$'s, for $i\in\{1,2,\cdots,n\}$. It is denoted by $wt_H(x)$. If $y\in \mathbb{F}_q^n$, the hamming distance is defined as hamming weight of vector $(x-y)$. 
	\end{defn}

	\begin{defn}
		The Hamming distance of a code $\mathcal{C}$, denoted by $d_H(\mathcal{C})$, is the number $d_H(\mathcal{C})=\text{min}\{wt_H(x) \text{ $|$ }x\neq0\}$.
	\end{defn}
	
	\begin{defn}
		Let $ r=a_1+a_2u+a_3v+a_4uv\in \mathcal{R} $, the Lee weight of vector $ r $, denoted by $wt_L(r)$, is defined as  $ wt_L(r)=wt_H(a_1,a_1+a_2,a_1+a_3,a_1+a_2+a_3+a_4) $. The definition of Lee weight can be extended to $\mathcal{R}^n$. Let  $s=(s_1,s_2,\dots,s_n) \in \mathcal{R}^n $, then the Lee weight of $s$ is  $wt_L(s)=\sum_{i=1}^{n} w_L(s_i)$. If $t=(t_1,t_2,\dots,t_n) \in  \mathcal{R}^n $ then the Lee distance between two vectors $ s $ and $ t $, denoted by $d_L(s,t)$ is defined as $ d_L(s,t)=wt_L(s-t)=\sum_{i=1}^{n}wt_L(s_i-t_i).$ 
	\end{defn}
	
	\begin{defn}
		The Lee distance of the  code $ \mathcal{C} $ is the number $ d_L(\mathcal{C})=\text{min}\{d_L(s-t) \text{ $|$ }s\neq t\}.$
	\end{defn}
	A function $ \rho : \mathcal{R} \mapsto \mathbb{F}_q^4 $ is said to be gray map, if it is bijective and distance preserving. The following  gray map is defined in \cite{skewc},
	\begin{align*}
		\rho:&\  \mathcal{R}\to \mathbb{F}_q^4\\
		\rho(r)=\rho(a_1+a_2u+a_3v+a_4uv)&=(a_1,a_1+a_2,a_1+a_3,a_1+a_2+a_3+a_4).
	\end{align*}
	An equivalent gray map for  $ r=\sum_{i=1}^{4} \gamma_ir_i\in \mathcal{R}$, where $ r_i\in \mathbb{F}_q $ for $i=1,2,3,4 $, is defined as
	$$\rho(r)=\rho\left( \sum_{i=1}^{4}\gamma_i r_i\right) =(r_1,r_2,r_3,r_4).$$
	In a similar way, we can easily extend this map from $\mathcal{R}^n$ to $\mathbb{F}_q^{4n}.$
	By the definition of gray map, $ \rho  $ is a  linear map over $ \mathbb{F}_q $ and it preserves distance  from $ (\mathcal{R}^n,d_L) $ to $ (\mathbb{F}_q^{4n},d_H) $ , where $d_L$ is Lee distance and $d_H$ is Hamming distance. The following result can be directly obtained from  the definition of $\rho$.
	\begin{proposition} 
		For a linear code $\mathcal{C}$ of length $n$ over the ring $\mathcal{R}$ with cardinality $q^k$ and  Lee distance  $d$, $ \rho(\mathcal{C}) $  is $[4n,k,d]$ linear code over $\mathbb{F}_q$.
	\end{proposition}
	
	Define a Frobenius  operator $ F $ over $\mathcal{R}$ as follows:
	\begin{align*}
		F&:\mathcal{R}\to \mathcal{R}\\
		F(a_1+a_2u+a_3v+a_4uv)&=a_1^p+ua_2^p+va_3^p+uva_4^p.
	\end{align*}

	Equivalently, for $r=\sum_{i=1}^{4}\gamma_i r_i\in \mathcal{R}$,	
	\begin{align*}
		F(r)&=\gamma_1r_1^p+\gamma_2r_2^p+\gamma_3r_3^p+\gamma_4r_4^p.
	\end{align*}

	For $ s=(s_1,s_2,\dots,s_n) $ and $  t=(t_1,t_2,\dots,t_n) \in \mathcal{R}^n$, define 
	$l$-Galois inner product, for $ 0\leq l\leq e-1 $, $$ [s,t]_l=\sum_{i=1}^{n}s_iF^{l}(t_i) .$$
	\begin{remark} 
		Note that, above inner product is generalization of  Euclidean and Hermitian inner product for $ l=0 $ and $ l=\frac{e}{2} $ (when $ e $ is even), respectively.
	\end{remark} 
	From onwards, we denote $ [s,t], [s,t]_H $ and $ [s,t]_l $ as Euclidean, Hermitian and $ l $-Galois inner product over $ \mathcal{R} $ and $ \langle s,t\rangle, \langle s,t\rangle_H $ and $ \langle s,t\rangle_l $ as Euclidean, Hermitian and $ l $-Galois inner product over $ \mathbb{F}_q $, respectively. The  $ l$-Galois dual code  $\mathcal{C}^{\perp_l}$ of $\mathcal{C}$ over $\mathcal{R}$ is defined as
	$$\mathcal{C}^{\perp_l}=\{s\in \mathcal{R}^n|\ \ [t,s]_l=0\ \ \forall t\in \mathcal{C}\}.$$
	Clearly, $\mathcal{C}^{\perp_l}$ is linear code over $\mathcal{R}$.
	A linear code over $ \mathcal{R} $ is said to be $ l $-Galois LCD if $ \mathcal{C}\cap \mathcal{C}^{\perp_l}=\{0\} .$ It is well known that for a Frobenius ring $ \mathcal{R} $ and a linear code $\mathcal{C}$ over the ring $ \mathcal{R} $ of length $ n $, the  product of cardinality of $\mathcal{C}$ and $\mathcal{C}^{\perp_l}$ is equal to  the cardinality of $ \mathcal{R}^n $  that is $ |\mathcal{C}||\mathcal{C}^{\perp_l}|=|\mathcal{R}^n| .$
	\begin{remark}
		For $ l=0 $ and $ l = e/2 $ (when $ e $ is even), it is  Euclidean  and  Hermitian dual code, respectively.
	\end{remark}
	
	\section{ $ { l } $-Galois Linear Codes over $ \mathbf{\mathcal{R}} $}  \label{coder}
	We  derive a necessary and sufficient condition for $\mathcal{C}$ to be an $ l$-Galois LCD code over $ \mathcal{R} $ with respect to its component codes. Also, we  give a relationship between $l$-Galois LCD codes and its gray images. 
	
	A linear code $\mathcal{C}$ over $ \mathcal{R} $ can be decomposed into  four  component codes over finite filed $ \mathbb{F}_q .$ Let us  introduce the component codes of $\mathcal{C}$ as follows
	\begin{align*}
		\mathcal{C}_1=\{x\in \mathbb{F}_q^n|\gamma_1x+\gamma_2y+\gamma_3z+\gamma_4w\in \mathcal{C}\ \text{for some}\ \ y,z,w\in \mathbb{F}_q^n\},\\
		\mathcal{C}_2=\{y\in \mathbb{F}_q^n|\gamma_1x+\gamma_2y+\gamma_3z+\gamma_4w\in \mathcal{C}\ \text{for some}\ \  x,z,w \in \mathbb{F}_q^n\},\\
		\mathcal{C}_3=\{z\in \mathbb{F}_q^n|\gamma_1x+\gamma_2y+\gamma_3z+\gamma_4w\in \mathcal{C}\ \text{for some} \ \ x,y,w \in \mathbb{F}_q^n\},\\
		\mathcal{C}_4=\{w\in \mathbb{F}_q^n|\gamma_1x+\gamma_2y+\gamma_3z+\gamma_4w\in \mathcal{C}\ \text{for some}\ \ x,y,z \in \mathbb{F}_q^n\}.
	\end{align*}
	It can be observed that $\mathcal{C}_i$'s are linear codes over $ \mathbb{F}_q $  for $ 1\leq i\leq 4$ and $\mathcal{C}=\gamma_1 \mathcal{C}_1\oplus\gamma_2 \mathcal{C}_2\oplus\gamma_3\mathcal{C}_3\oplus\gamma_4 \mathcal{C}_4$. We say $ \mathcal{C}_1,\mathcal{C}_2,\mathcal{C}_3 $ and $ \mathcal{C}_4 $ are component codes of the linear code $ \mathcal{C} .$  The cardinality of a linear code $\mathcal{C}$ is product of cardinalities of its component codes, 
	$$|\mathcal{C}|=|\mathcal{C}_1||\mathcal{C}_2||\mathcal{C}_3||\mathcal{C}_4|.$$
	The Lee distance of a linear code $\mathcal{C}$ is the minimum of Hamming distances of component codes $ \mathcal{C}_i $'s, 
	$$d_L(\mathcal{C})=\min_{1\leq i\leq 4}\{d_H(\mathcal{C}_i)\}.$$ 
	Let us define $\mathcal{C}^{p^l}=\{(F^l(c_1),F^l(c_2),...,F^l(c_n))| (c_1,c_2,...,c_n)\in \mathcal{C}\}$  and  $F^l(G)=(F^l(g_{ij}))$ for any matrix $ G=(g_{ij}) $ over $ \mathcal{R} $. For $ 1\leq i \leq 4 $, let $G_i$ be generator matrix for $\mathcal{C}_i$, then the generator matrix  for $\mathcal{C}$ is
	$$G=\begin{bmatrix}
		\gamma_1G_1\\
		\gamma_2G_2\\
		\gamma_3G_3\\
		\gamma_4G_4
	\end{bmatrix}.$$
	The generator matrix for $\rho(\mathcal{C})$ is
	
	$$\rho(G)=\begin{bmatrix}
		\rho(\gamma_1G_1)\\
		\rho(\gamma_2G_2)\\
		\rho(\gamma_3G_3)\\
		\rho(\gamma_4G_4)
	\end{bmatrix},$$
	where $\rho(\gamma_iG_i)$ is matrix over $\mathbb{F}_q$ for $ 1\leq i\leq 4. $  Since $ \gamma_i\gamma_j=0 $ for $ i\neq j $ and $ \gamma_i^2=\gamma_i $ for $ i=1,2,3,4 ,$  we get
	\begin{align*}
		G(F^{e-l}(G))^T=&\begin{bmatrix}
			\gamma_1 G_1F^{e-l}(G_1)^T & 0 &0 &0\\
			0&\gamma_2 G_2F^{e-l}(G_2)^T &0&0\\
			0&0&\gamma_3 G_3F^{e-l}(G_3)^T&0\\
			0&0&0&\gamma_4 G_4F^{e-l}(G_4)^T
		\end{bmatrix}.
	\end{align*} 
	We call $\mathcal{C}$ as  an $ [n,k,d] $ code over $ \mathcal{R} $ if $\mathcal{C}$ is the  code of length $ n $,  $ |\mathcal{C}|=q^k $ and $ d $ is the Lee distance. If $\mathcal{C}_i $'s are component codes  with parameters $ [n,k_i,d_i] $ for $ i=1,2,3,4 $ then   $ k=\sum\limits_{i=1}^4k_i $ and $ d= \min\limits_{1\leq i\leq 4}\{d_i\} .$
	
	In the following lemma, we observe that the Euclidean dual of $\mathcal{C}^{p^{(e-l)}}$ is equal to the $l$-Galois dual code of $\mathcal{C}$.
	\begin{lemma}\label{lemma3.1}
		If $\mathcal{C}$ is an $ [n,k,d] $  linear code  over $ \mathcal{R} $ then  $ \mathcal{C}^{p^{(e-l)}} $ is an  $[n,k,d]$ linear code over $ \mathcal{R} $ and $ \mathcal{C}^{\perp_l}=\left(\mathcal{C}^{p^{(e-l)}}\right) ^\perp .$ Moreover, if $\mathcal{C}$ has a generator matrix $G$, then a generator matrix of $\mathcal{C}^{p^{(e-l)}}$ is  $ F^{e-l}(G)$. 
	\end{lemma}
	Next theorem provides the decomposition of the $l$-Galois dual code into its component codes. Consequently, we obtain a relation between  $l$-Galois LCD codes over $ \mathcal{R} $ and  $l$-Galois LCD of its component codes over $ \mathbb{F}_q .$      
	\begin{theorem}\label{bthm}
		The  following holds for a linear code  $ \mathcal{C}= \bigoplus_{i=1}^4\gamma_i \mathcal{C}_i $  over $ \mathcal{R} $.
		\begin{enumerate}
			\item \label{b1} $ \mathcal{C}^{\perp_l}= \bigoplus_{i=1}^4\gamma_i \mathcal{C}_i^{\perp_l} .$
			\item \label{b2}  $\mathcal{C}$ is an $ l $-Galois LCD code over $ \mathcal{R} $ if and only if its all component codes $ \mathcal{C}_i $'s are $ l $-Galois LCD code over $ \mathbb{F}_q $ for $ 1\leq i\leq 4. $
			\item \label{b3} $\mathcal{C}$ is an $ l $-Galois self orthogonal linear code over $\mathcal{R}$ if and only if its all component codes $ \mathcal{C}_i $'s are $ l $-Galois self orthogonal codes over $ \mathbb{F}_q .$ Also, $\mathcal{C}$ is a self-dual code if and only if its all component codes $ \mathcal{C}_i $'s are	self-dual codes over $ \mathbb{F}_q  $ for $ 1\leq i\leq 4. $ 
		\end{enumerate}
	\end{theorem}
	\begin{proof}
		\begin{enumerate}
			\item  Let $ x=\gamma_1x_1+\gamma_2x_2+\gamma_3x_3+\gamma_4x_4\in \mathcal{C}^{\perp_l}$ then $ [y,x]_l=0 $ for all $ y=\gamma_1y_1+\gamma_2y_2+\gamma_3y_3+\gamma_4y_4\in \mathcal{C}.$ Since $ \gamma_i^2=\gamma_i$ and $ \gamma_i\gamma_j=0 $ for $ i\neq j,\  [y,x]_l=\gamma_1\langle y_1,x_1\rangle_l+\gamma_2\langle y_2,x_2\rangle_l+\gamma_3\langle y_3,x_3\rangle_l+\gamma_4\langle y_4,x_4\rangle_l $ implies that $ \langle y_i,x_i\rangle_l=0 $ for all $ y_i\in \mathcal{C}_i $ and $ i=1,2,3,4 $ i.e.,  $ x_i\in \mathcal{C}_i^{\perp_l} $ for $ i=1,2,3,4 .$ Therefore, $ x\in \gamma_1 \mathcal{C}_1^{\perp_l}\oplus\gamma_2 \mathcal{C}_2^{\perp_l}\oplus\gamma_3 \mathcal{C}_3^{\perp_l}\oplus\gamma_4 \mathcal{C}_4^{\perp_l}  .$\\
			Conversely, let $ w=\gamma_1w_1+\gamma_2w_2+\gamma_3w_3+\gamma_4w_4\in  \gamma_1 \mathcal{C}_1^{\perp_l}\oplus\gamma_2 \mathcal{C}_2^{\perp_l}\oplus\gamma_3 \mathcal{C}_3^{\perp_l}\oplus\gamma_4 \mathcal{C}_4^{\perp_l} $, where $ w_i\in \mathcal{C}_i^{\perp_l} $. For any $ y=\gamma_1y_1+\gamma_2y_2+\gamma_3y_3+\gamma_4y_4\in \mathcal{C}$, where $ y_i\in \mathcal{C}_i$,\  $[y,w]_l=\gamma_1\langle y_1,w_1\rangle_l+\gamma_2\langle y_2,w_2\rangle_l+\gamma_3\langle y_3,w_3\rangle_l+\gamma_4\langle y_4,w_4\rangle_l=0 $  implies that $ w\in \mathcal{C}^{\perp_l}.$ Hence  $ \mathcal{C}^{\perp_l}= \gamma_1 \mathcal{C}_1^{\perp_l}\oplus\gamma_2 \mathcal{C}_2^{\perp_l}\oplus\gamma_3 \mathcal{C}_3^{\perp_l}\oplus\gamma_4 \mathcal{C}_4^{\perp_l} .$ 
			\item Let us  suppose that $\mathcal{C}$ is an $ l $-Galois LCD code over $ \mathcal{R} $ that is $ \mathcal{C}\cap \mathcal{C}^{\perp_l}=\{0\}$.  Let $ x_i\in \mathcal{C}_i\cap\mathcal{C}_i^{\perp_l} $  for some $ i=1,2,3,4 $  that is $ \langle y_i,x_i\rangle_l=0  $ for all $ y_i\in \mathcal{C}_i .$
			Now take $ x=\gamma_ix_i \in \mathcal{C}$. Then for any $ y=\gamma_1y_1+\gamma_2y_2+\gamma_3y_3+\gamma_4y_4\in \mathcal{C}$, where $\ y_j\in \mathcal{C}_j \ \forall j=1,2,3,4 $, $ [y,x]_l=[\gamma_1y_1+\gamma_2y_2+\gamma_3y_3+\gamma_4y_4,\gamma_ix_i]_l=\gamma_i\langle y_i,x_i \rangle_l=0$, since $ \gamma_i^2=\gamma_i$ and $ \gamma_i\gamma_j=0 $ for $i\neq j .$ This implies that $ x\in\mathcal{C}\cap \mathcal{C}^{\perp_l}=\{0\} $ i.e., $ x=0 $, consequently, $ x_i=0 .$ Hence $ \mathcal{C}_i $ is the $ 
			l $- Galois LCD code over $ \mathbb{F}_q .$\\
			Conversely, suppose $ \mathcal{C}_i $'s are $ l $-Galois LCD  code over $ \mathbb{F}_q $ for $ i=1,2,3,4 .$ Let $ x=\gamma_1x_1+\gamma_2x_2+\gamma_2x_3+ \gamma_4x_4\in \mathcal{C}\cap \mathcal{C}^{\perp_l}$ then  $ x_i $'s in $ \mathcal{C}_i\cap \mathcal{C}_i^{\perp_l} $ and $\mathcal{C}_i\cap \mathcal{C}_i^{\perp_l}=\{0\} $ implies that $ x=0. $ Thus $\mathcal{C}$ is $ l $-Galois LCD code. 
		\end{enumerate}
The	proof of \ref{b3} can be easily followed by \ref{b1}, so we are omitting the proof. 
	\end{proof}
	\begin{remark}
		Result \ref{b1} and \ref{b3} in above Theorem \ref{bthm} have been done for the Euclidean dual over the ring $ \mathbb{Z}_4 +u\mathbb{Z}_4+v\mathbb{Z}_4+uv\mathbb{Z}_4$ in \cite{someresultsz4}. 	
	\end{remark}
	The  next corollary gives a necessary and sufficient condition on generator matrices  for $ l $-Galois LCD code over $\mathcal{R}.$ 
	\begin{corollary}
		Let $ \mathcal{C}=\bigoplus_{i=1}^4\gamma_i \mathcal{C}_i$ with generator matrix 
		$$G=\begin{bmatrix}
			\gamma_1G_1\\
			\gamma_2G_2\\
			\gamma_3G_3\\
			\gamma_4G_4
		\end{bmatrix},$$
		where $ G_i $ is generator matrix for $ \mathcal{C}_i $ over $ \mathbb{F}_q .$ Then $\mathcal{C}$ is an $ l $-Galois LCD code over $ \mathcal{R} $ if and only if $ G_i(F^{e-l}(G_i))^T $ is non singular matrix for $ i=1,2,3,4 $ over $ \mathbb{F}_q .$ 
	\end{corollary}  
	\begin{proof}
		By Theorem \ref{bthm}, $\mathcal{C}$ is an $ l $-Galois LCD code if and only if $ \mathcal{C}_i $'s are $ l $-Galois LCD  code over $ \mathbb{F}_q.$ From \cite[Theorem 2.4]{galoislcd}  , $ \mathcal{C}_i $'s are $ l $-Galois LCD code iff $ G_i(F^{e-l}(G_i))^T $ is non singular matrix over $ \mathbb{F}_q .$
	\end{proof}
	Now, by using the definition of $\rho$ we derive some useful properties on gray image of  $ l $-Galois dual codes over $ 
	\mathcal{\mathcal{R}} $.
	\begin{lemma}
		If $ \mathcal{C} $ is an $ [n,k] $  linear code  over $ \mathcal{R} $ then
		$\rho(\mathcal{C}^{\perp_l})={\rho(\mathcal{C})}^{\perp_l}.$
		\begin{proof} Let $\rho(x)\in\rho(\mathcal{C}^{\perp_l}) $, where $x\in \mathcal{C}^{\perp_l}$ this implies  $[z,x]_l=0 \ \forall z \in \mathcal{C}$. Let $z=\sum_{i=1}^{4}\gamma_iz_i$ and $x=\sum_{i=1}^{4}\gamma_ix_i$, where $z_i\in \mathcal{C}_i$ and $x_i\in \mathcal{C}_i^{\perp_l}$,\ $[z,x]_l=\gamma_1\left\langle z_1,x_1\right\rangle_l+\gamma_2\left\langle z_2,x_2\right\rangle_l+\gamma_3\left\langle z_3,x_3\right\rangle_l+\gamma_4\left\langle z_4,x_4\right\rangle_l =0. $ Hence $\left\langle z_i,x_i\right\rangle_l=0 $ for $i=1,2,3,4.$	Now, $\left\langle \rho(z),\rho(x)\right\rangle _l=\sum_{i=1}^{4}z_i\cdot x_i^{p^l}=\sum_{i=1}^{4}\left\langle z_i,x_i\right\rangle_l=0 \ \ \forall  \rho(z)\in \rho(\mathcal{C})$ this implies $\rho(x)\in  {\rho(\mathcal{C})}^{\perp_l}$.
			Therefore, $\rho(\mathcal{C}^{\perp_l})\subseteq{\rho(\mathcal{C})}^{\perp_l}.$
			
			Conversely, the cardinality of $\rho(\mathcal{C}^{\perp_l})$ equal to $\mathcal{C}^{\perp_l}$ that is $|\rho(\mathcal{C}^{\perp_l})|=\frac{q^{4n}}{|\mathcal{C}|}.$ Moreover, $|\rho(\mathcal{C})^{\perp_l}|=\frac{q^{4n}}{|\rho(\mathcal{C})|}=\frac{q^{4n}}{|\mathcal{C}|}$.
			Hence $ \rho(\mathcal{C}^{\perp_l})={\rho(\mathcal{C})}^{\perp_l}. $
		\end{proof}
	\end{lemma}	
	\begin{lemma}\label{lemma 2.3}	
		If $ \mathcal{C} $ is a linear code over $ \mathcal{R} $ then
		$\rho(\mathcal{C}\cap \mathcal{C}^{\perp_l})=\rho(\mathcal{C})\cap\rho(\mathcal{C}^{\perp_l}).$
		\begin{proof}
			Let $ \rho(x)\in\rho(\mathcal{C}\cap \mathcal{C}^{\perp_l})$ for some $x\in \mathcal{C}\cap \mathcal{C}^{\perp_l}$ this implies  $\rho(x)\in\rho(\mathcal{C})\cap\rho(\mathcal{C}^{\perp_l})$. Hence $ \rho(\mathcal{C}\cap \mathcal{C}^{\perp_l})\subseteq\rho(\mathcal{C})\cap\rho(\mathcal{C}^{\perp_l}).$
			Conversly, let $ y \in \rho(\mathcal{C})\cap\rho(\mathcal{C}^{\perp_l}) $ this implies $ y\in \rho(\mathcal{C}) \ \text{and} \  y\in \rho(\mathcal{C}^{\perp_l}) $. Since $\rho$ is bijective function so there is unique  $x\in \mathcal{C}\cap \mathcal{C}^{\perp_l}$ such that $\rho(x)=y$. Hence $  \rho(\mathcal{C})\cap\rho(\mathcal{C}^{\perp_l})\subseteq \rho(\mathcal{C}\cap \mathcal{C}^{\perp_l}) $.
			Therefore, $\rho(\mathcal{C}\cap \mathcal{C}^{\perp_l})=\rho(\mathcal{C})\cap\rho(\mathcal{C}^{\perp_l}).$
			
		\end{proof}
	\end{lemma}		
	\begin{theorem}\label{gthm}
		A linear code $\mathcal{C}$ is an $ l $-Galois LCD code over $ \mathcal{R} $ if and only if $ \rho(\mathcal{C})$ is an $ l $-Galois LCD code over $ \mathbb{F}_q .$
	\end{theorem}
	\begin{proof}
		Suppose $\mathcal{C}$ is an $ l $-Galois LCD code over $ \mathcal{R} $ that is  $\mathcal{C}\cap \mathcal{C}^{\perp_l}=\{0\}$. From  Lemma \ref{lemma 2.3}, we get $\rho(\mathcal{C})\cap\rho(\mathcal{C})^{\perp_l}=\{0\}$. Conversely, if $ \rho(\mathcal{C})$ is $ l $-Galois LCD code then $$\{0\}=\rho(\mathcal{C})\cap\rho(\mathcal{C})^{\perp_l}=\rho(\mathcal{C})\cap\rho(\mathcal{C}^{\perp_l})=\rho(\mathcal{C}\cap \mathcal{C}^{\perp_l})$$
		this implies $\mathcal{C}\cap \mathcal{C}^{\perp_l}=\{0\}$. Therefore, $\mathcal{C}$ is an $ l $-Galois LCD code over $ \mathcal{R}. $
	\end{proof}
	
	\section{ Construction of  Galois LCD code  equivalent to linear code} \label{construction}
	Here, we discuss a construction of the Euclidean and the $ l $-Galois LCD codes over $ \mathcal{R} $ with the help of its components codes over $\mathbb{F}_q$.  We have shown that for every linear code $\mathcal{C}$ there exists  a Euclidean LCD code and an $ l $-Galois LCD code which are equivalent to $ \mathcal{C} .$  
	
	For two integers $ m\geq1 $ and $ 0\leq w\leq m $, let $b$  be an element in $ \mathbb{F}_q^m $ with the Hamming weight $ w .$  Define  support of $b$ is the set, $ S= \{i_1,i_2,\dots, i_w\} ,$ of indices where components of $b$ are non zero. Denote $ m\times m $ diagonal matrix whose entries are $ b_1,b_2,\dots , b_m $ by $ diag_m[b] .$ Let $ P  $ be an $ m\times m $ square matrix over $ \mathbb{F}_q .$ Define $ P_S $ as the submatrix of $ P $ obtained by deleting the $ i_1,i_2,\dots,i_w $-th columns and rows of $ P .$ We write $ P_S=I $  if $ S=\{1,2,\dots ,m\} $  and $ P_{\emptyset} =P.$ Then from \cite{cafinite}, we have the following lemma.
	\begin{lemma}\label{ldet}
		For a matrix $ P $ of order $ m\times m $ over $ \mathbb{F}_q $ and an integer $ t $ such that $ 0\leq t\leq m-1 .$ Assume that for any $ J\subseteq \{1,2,\dots ,m\} $ with $ 0\leq| J|\leq t ,\  det(P_J)=0 .$ Then for every element $b \in \mathbb{F}_q^m$ with support $S$ and  Hamming weight $ w $ such that $ 1\leq w\leq t+1,$  we have
		$$det \left( P+diag_m[b]\right) =\left(\prod_{i\in S}b_i\right) det(P_S).$$    
	\end{lemma}
	Suppose  $ \alpha=(\alpha_1,\alpha_2,\dots , \alpha_n)\in \mathcal{R}^n$ is a fixed element, where $ \alpha_j=\sum_{i=1}^{4}\gamma_i\alpha_{ji}, \ \alpha_{ji}\in\mathbb{F}_q$ for $ j=1,2,\dots,n .$ Define $$ \mathcal{C}^{\alpha}=\{\alpha\cdot c|\ c\in \mathcal{C}\}=\{(\alpha_1c_1,\alpha_2c_2,\dots, \alpha_nc_n)| \ (c_1,c_2,\dots,c_n)\in \mathcal{C}\} .$$ Clearly, $ \mathcal{C}^{\alpha} $ is a linear code over $ \mathcal{R} .$
	Let $ G=\begin{bmatrix}
		\gamma_1G_1\\
		\gamma_2G_2\\
		\gamma_3G_3\\
		\gamma_4G_4
	\end{bmatrix} $ be a generator matrix for $\mathcal{C}$ then  $ G^{\alpha}=\begin{bmatrix}
		\gamma_1G_1^{\alpha_1'}\\
		\gamma_2G_2^{\alpha_2'}\\
		\gamma_3G_3^{\alpha_3'}\\
		\gamma_4G_4^{\alpha_4'}
	\end{bmatrix} $ is a generator matrix for $ \mathcal{C}^{\alpha} $  obtained by multiplying $ j $-th column of $ G $ by $ \alpha_j $, where $ G_i^{\alpha_i'} $ is the matrix obtained by multiplying $ j $-th column of $ G_i $ by $ \alpha_{ji} $ and $\alpha_i'=(\alpha_{1i},\alpha_{2i},\dots,\alpha_{ni})\in \mathbb{F}_q^n $ for $ i=1,2,3,4. $
	\begin{remark}
		Note that, for $\alpha=(\alpha_1,\alpha_2,\dots , \alpha_n)\in \mathcal{R}^n$ such that $ \alpha_j\neq 0 $ for each $ 1\leq j\leq n $, $\mathcal{C}$ and $ \mathcal{C}^{\alpha} $ are equivalent codes over $ \mathcal{R} .$
	\end{remark}
	The following theorem  gives a construction of  Euclidean LCD codes over $ \mathcal{R} $ from linear codes over $ \mathcal{R} .$ We denote the parameters of component codes $ C_i $ by $ [n,k_i,d_i] $ for $ i=1,2,3,4. $
	\begin{theorem}\label{thm 4.2}
		All notations are as above. Let $ \mathcal{C}=\bigoplus_{i=1}^{4}\gamma_i\mathcal{C}_i$ be an $ [n,k,d] $ linear code over $ \mathcal{R} ,$ where $ \mathcal{C}_i $'s are component codes over $ \mathbb{F}_q $ with generator matrices $ G_i=[I_{k_i}:M_i] $. Let $ P_i=G_iG_i^T$ and $ t_i\leq k_i-1 $ be non-negative integers such that $ det((P_i)_{S_i})=0 $ for any  $ S_i \subseteq \{1,2,\dots, k_i\} $ with $ 0\leq |S_i|\leq t_i $ and assume there exist $ R_i\subseteq\{1,2,\dots k_i\} $ with cardinality $ t_i+1 $ such that $ det((P_i)_{R_i}) \neq 0.$ Suppose $\alpha\in \mathcal{R}^n$   such that $ \alpha_{ji}\in \mathbb{F}_q\setminus\{+1,-1\} $ if  $ j\in R_i $ and $ \alpha_{ji}\in \{+1,-1\} $ if  $ j\in\{1,2,\dots n\}\setminus R_i, $ for $\ i=1,2,3,4.$ Then $ \mathcal{C}^{\alpha} $ is a Euclidean LCD code over $ \mathcal{R} .$
	\end{theorem}	
	\begin{proof}
		Let $\alpha=(\alpha_1,\alpha_2,\dots , \alpha_n)\in \mathcal{R}^n$ and $ c=(c_1,c_2,\dots,c_n)\in \mathcal{C} .$	Note that, for $\alpha_j,c_j\in \mathcal{R}$, we write $ \alpha_j=\sum_{i=1}^{4}\gamma_i\alpha_{ji}, \ \alpha_{ji}\in\mathbb{F}_q$ and $ c_j=\sum_{i=1}^{4} \gamma_i c_{ji} ,\ c_{ji}\in \mathcal{C}_i$  then $ \alpha_jc_j=\left(\sum_{i=1}^{4}\gamma_i\alpha_{ji}\right)\left(\sum_{i=1}^{4} \gamma_i c_{ji}\right)=\sum_{i=1}^{4}\gamma_i\alpha_{ji}c_{ji} $, since $ \gamma_i\gamma_m=0 $ for $ i\neq m $ and $ \gamma_i^2=\gamma_i $ for $ i,m=1,2,3,4 $ and $j=1,2,\dots,n.$ Now,
		\begin{align*}
			\mathcal{C}^{\alpha}=&\{(\alpha_1c_1,\alpha_2c_2,\dots, \alpha_nc_n)| \ (c_1,c_2,\dots,c_n)\in \mathcal{C}\}\\
			=&\left\lbrace \left(\sum_{i=1}^{4}\gamma_i \alpha_{1i}c_{1i},\sum_{i=1}^{4}\gamma_i \alpha_{2i}c_{2i},\dots ,\sum_{i=1}^{4}\gamma_i \alpha_{ni}c_{ni}\right)\big\rvert \  (c_1,c_2,\dots,c_n)\in \mathcal{C} 
			\right\rbrace  \\
			=&\left\lbrace \sum_{i=1}^{4}\gamma_i(\alpha_{1i}c_{1i},\alpha_{2i}c_{2i},\alpha_{3i}c_{3i},\dots, \alpha_{ni}c_{ni})\ \big\rvert  \  (c_1,c_2,\dots,c_n)\in \mathcal{C}
			\right\rbrace  \\
			=&\bigoplus_{i=1}^{4}\gamma_i \mathcal{C}_i^{\alpha_i'}   .
		\end{align*}
		Where $\mathcal{C}_i^{\alpha_i'}=\{(\alpha_{1i}c_{1i},\alpha_{2i}c_{2i},\alpha_{3i}c_{3i},\dots, \alpha_{ni}c_{ni})| \ (c_{1i},c_{2i},\dots,c_{ni})\in \mathcal{C}_i\}$ and $\alpha_i'=(\alpha_{1i},\alpha_{2i},\dots,\alpha_{ni})$. It is clear that $\mathcal{C}_i^{\alpha_i'}$'s are linear codes over $\mathbb{F}_q$ with generator matrices $ G_i^{\alpha_i'} $ for $i=1,2,3,4.$
		Also, $ \alpha=\sum_{i=1}^{4}\gamma_i \alpha_i' .$ 	 
		From  \cite[Theorem 11]{cafinite} , each $ \mathcal{C}_i^{\alpha_i'} $ is Euclidean LCD code over $ \mathbb{F}_q $ and by Theorem \ref{bthm} (take $ l=0 $) , $ \mathcal{C}^{\alpha} $ is a Euclidean LCD code.
	\end{proof}
	Next, we use the technique described in \cite{cafinite} to establish the existence of  $ \alpha $ for which $ \mathcal{C}^{\alpha} $ is a Euclidean LCD  code for a given linear code $\mathcal{C}$ over $ \mathcal{R}. $ 	
	\begin{corollary}\label{Eequiv}
		Let  $ \mathbb{F}_q\ (q>3) $ be a finite field and $\mathcal{C}$ be an $ [n,k,d] $ linear code over $ \mathcal{R} .$ Then $ \mathcal{C}^{\alpha} $ is an $ [n,k,d] $ Euclidean LCD code over $ \mathcal{R} $ for some $ \alpha =(\alpha_1,\alpha_2,\dots , \alpha_n)$ in $ \mathcal{R}^n $
		with $ \alpha_j\neq 0 $ for  $ 1\leq j\leq n .$
	\end{corollary}
	\begin{proof}
		Let $ \mathcal{C}=\bigoplus_{i=1}^4\gamma_i\mathcal{C}_i $ be a linear code over $ \mathcal{R} $. If $\mathcal{C}$ is a  Euclidean LCD code then we can take $ \alpha=(\alpha_1,\alpha_2,\dots , \alpha_n) \in \mathcal{R}^n$ such that $ \alpha_j= \gamma_1+\gamma_2 +\gamma_3+\gamma_4$ for $ 1\leq j\leq n. $	Then $ \mathcal{C}^{\alpha}=\mathcal{C} $, Euclidean LCD code over $ \mathcal{R} .$
		
		For the case when $\mathcal{C}$ is not a Euclidean LCD code then by Theorem \ref{bthm}, for some $ i=1,2,3,4,\   \mathcal{C}_i $ is not Euclidean LCD. Let $ G_i  $ be generator matrix for $ \mathcal{C}_i $ then $ det(G_iG_i^T)=0 .$ Suppose $ P_i =G_iG_i^T$, then there exists an integer $ t_i\geq0 $ and  $ R_i \subseteq \{1,2,\dots, k_i \} $ with cardinality $ | R_i|=t_i+1 $ such that $ det((P_i)_{R_i})\neq 0 $ and $ det((P_i)_{S_i}) =0$ for any $ S_i \subseteq \{1,2,\dots ,k_i\} $ with  $ 0\leq |S_i|\leq t_i $. Also, $ \mathbb{F}_q^*\setminus \{-1,1\}\neq \emptyset $ since $ q>3 $. Choose $ \alpha_i'=(\alpha_{1i},\alpha_{2i},\dots,\alpha_{ni}) \in \mathbb{F}_q^n $ such that $ \alpha_{ji} \in  \mathbb{F}_q^*\setminus \{-1,1\} $ if $ j\in R_i $ and $ \alpha_{ji}=1 $ if $ j\in \{1,2,\dots , k_i\}\setminus R_i.$ Then by  \cite[Theorem 11]{cafinite}, $ \mathcal{C}_i^{\alpha_i'} $ is a Euclidean LCD code over $ \mathbb{F}_q .$  Take $ \alpha=\sum_{m=1}^4 \gamma_m \alpha_m' \in \mathcal{R}^n $, where $ \alpha_m'=\alpha_i' $ for $ m=i $ and $ \alpha_m'=(1,1,\dots,1) $ for $ m\neq i .$ By Theorem \ref{thm 4.2}$,  \mathcal{C}^{\alpha}=\bigoplus_{m=1}^4 \gamma_m\mathcal{C}_m^{\alpha_m'} $ is an  $ [n,k,d ] $ Euclidean  LCD code over $ \mathcal{R} .$ 
	\end{proof}
 In upcoming theorem, we  construct  an $ l $-Galois LCD code from a given linear code over a finite field $ \mathbb{F}_q .$  Then similarly, we provide the construction over  $ \mathcal{R}.$ 
	\begin{theorem}\label{kgaloisfield}
		Let  $\mathbb{F}_q \ (q=p^e)$ be a finite field. For $0<l<e$ and $p^{e-l}+1\mid p^e-1,  \beta =\frac{p^e-1}{p^{e-l}+1}$ (say). Let  $G=[I_k:M]$ be a generator matrix for a linear code $\mathcal{C}$ over $ \mathbb{F}_q $ with parameters $[n,k,d]$ and denote the matrix $G{\left( F^{e-l}(G)\right) }^T$ by $ P $. Let $ 0\leq t\leq k-1$ be an integer  for which any  $I\subseteq\{1,2,...,k\}$ with $0\leq | I|\leq t$, $det(P_I)=0$  and assume there exist $J\subseteq \{1,2,...,k\} $ with cardinality $t+1$ such that $ det(P_J)\neq 0 $. Suppose $ a \in \mathbb{F}_q^n$ such that  $ a_j\in \mathbb{F}_q \setminus (\mathbb{F}_q^*)^{\beta}$ for $ j\in J $ and   $ a_j\in(\mathbb{F}_q^*)^{\beta}$ for $ j\in \{1,2,\dots,n\}\setminus J $. Then $ \mathcal{C}^{a} $ is an $ l $-Galois LCD code over $ \mathbb{F}_q .$
	\end{theorem}
	\begin{proof}
		A generator matrix $ G^a $ for $ \mathcal{C}^a $ is obtained by multiplying $ j $-th column of matrix $ G $ by $ a_j $ for $ 1\leq j\leq n .$ Then $ G^{a}(F^{e-l}(G^{a}))^T=G{\left( F^{e-l}(G)\right) }^T+ \text{diag}_k[b] $, where $ b=(a_1^{p^{e-l}+1}-1,a_2^{p^{e-l}+1}-1,\dots , a_k^{p^{e-l}+1}-1) .$ Note that support of $ b $ is the set $ J $. By Lemma \ref{ldet}  $ det(G^{a}(F^{e-l}(G^{a}))^T)=det\left( P+ diag_k[b]\right) =\left(\prod_{j\in J}b_j \right) det(P_J)\neq 0  .$ Hence $ \mathcal{C}^{a} $ is an $ l $-Galois LCD code over $ \mathbb{F}_q .$
	\end{proof}	
	\begin{theorem}\label{thm 4.4}
		All notations are as above.  For $0<l<e$ and $p^{e-l}+1\mid p^e-1,\ \  \beta =\frac{p^e-1}{p^{e-l}+1}$ (say). Let $ \mathcal{C}=\bigoplus_{i=1}^{4}\gamma_i\mathcal{C}_i$ be an $ [n,k,d] $ linear code over $ \mathcal{R} ,$ where $ \mathcal{C}_i$'s are component codes over $ \mathbb{F}_q $ with generator matrices $ G_i$'s. Let $ P_i=G_i{\left(F^{e-l}( G_i)\right) }^T $ and $ 0\leq t_i\leq k_i-1 $ be an integer such that $ det((P_i)_{S_i})=0 $ for any  $S_i\subseteq \{1,2,\dots, k_i\} $ with $ 0\leq |S_i|\leq t_i $ and assume there exist $ R_i\subseteq\{1,2,\dots k_i\} $ with cardinality $ t_i+1 $ such that $ det((P_i)_{R_i}) \neq 0.$ Let $\alpha\in \mathcal{R}^n$   such that $ \alpha_{ji}\in \mathbb{F}_q\setminus(\mathbb{F}_q^*)^{\beta}$ for all $ j\in R_i $ and $ \alpha_{ji}\in (\mathbb{F}_q^*)^{\beta} $ for all $ j\in\{1,2,\dots n\}\setminus R_i,$ for $ i=1,2,3,4.$ Then $ \mathcal{C}^{\alpha} $ is an $l$-Galois LCD code over $ \mathcal{R} .$
	\end{theorem}
	\begin{proof}
		Since $\mathcal{C}^{\alpha}=\bigoplus_{i=1}^{4}\gamma_i \mathcal{C}_i^{\alpha_i'}, \ \text{where} \  \alpha_i'=(\alpha_{1i},\alpha_{2i},\dots,\alpha_{ni})\in \mathbb{F}_q^n $ and $\mathcal{C}_i^{\alpha_i'}$'s are linear codes over $\mathbb{F}_q$ with generator matrices $ G_i^{\alpha_i'}.$  By the above Theorem \ref{kgaloisfield}, $\mathcal{C}_i^{\alpha_i'}$ are $ l $-Galois LCD codes over $\mathbb{F}_q$ for $i=1,2,3,4$. Therefore, by Theorem \ref{bthm}, $\mathcal{C}^{\alpha}$ is an $ l $-Galois LCD code over $\mathcal{R}$.
	\end{proof}
	The following corollary shows the existence of $\alpha$ for which $\mathcal{C}^{\alpha}$ is an $l$-Galois LCD code, equivalent to linear code $\mathcal{C}$ over $\mathcal{R}$.
	\begin{corollary}\label{kequiv}
		Let  $ \mathbb{F}_q \ (q=p^e)$  be a finite field. For $0<l<e$ and $p^{e-l}+1\mid p^e-1,\ \  \beta =\frac{p^e-1}{p^{e-l}+1}$ ($ \beta >1 $). Let $\mathcal{C}$ be an $ [n,k,d] $ linear code over the ring $ \mathcal{R} .$ Then $ \mathcal{C}^{\alpha} $ is an $ [n,k,d] $ $ l $-Galois LCD code over the ring $ \mathcal{R} $ for some  $\alpha=(\alpha_1,\alpha_2,\dots , \alpha_n)\in \mathcal{R}^n$ with $ \alpha_j\neq 0 $ for  $ 1\leq j\leq n .$
	\end{corollary}
	\begin{proof}
		Let $ \mathcal{C}=\bigoplus_{i=1}^4\gamma_i\mathcal{C}_i $ be a linear code over $ \mathcal{R} $. Take $ \alpha=(\alpha_1,\alpha_2,\dots , \alpha_n) \in \mathcal{R}^n$, where $ \alpha_j= \gamma_1+\gamma_2 +\gamma_3+\gamma_4$ for $ 1\leq j\leq n $,  if $\mathcal{C}$ is $ l $-Galois LCD code over $ \mathcal{R} .$ Then $ \mathcal{C}^{\alpha}=\mathcal{C} $, which is  $ l $-Galois LCD code over $ \mathcal{R} .$
					
		If $\mathcal{C}$ is not $ l $-Galois LCD code then by Theorem \ref{bthm}, for some $ 1\leq i\leq 4,\   \mathcal{C}_i $ is not $ l $-Galois LCD code. Let $ G_i $ be generator matrix for $ \mathcal{C}_i $ then  $ det(G_i(F^{e-l}(G_i))^T)=0 .$ Let $ P_i=G_i(F^{e-l}(G_i))^T $, then there exists an integer $ t_i\geq0 $ and  $ R_i \subseteq \{1,2,\dots k_i\} $ with cardinality $ | R_i| =t_i+1$  such that  $ det((P_i)_{R_i})\neq 0 $ and  $ det((P_i)_{S_i})=0 $ for any  $ S_i \subseteq \{1,2,\dots k_i\} $ with cardinality $ 0\leq  |S_i| \leq t_i .$  Also, $  \mathbb{F}_q^*\setminus (\mathbb{F}_q^*)^\beta \neq \emptyset $ since $ \beta>1.$   Choose $ \alpha_i'=(\alpha_{1i},\alpha_{2i},\dots,\alpha_{ni}) \in \mathbb{F}_q^n $ such that $ \alpha_{ji} \in  \mathbb{F}_q^*\setminus (\mathbb{F}_q^*)^\beta $ for $ j\in R_i $ and $ \alpha_{ji}=1 $ for $ j\in \{1,2,\dots , k_i\}\setminus R_i.$ Then by Theorem \ref{kgaloisfield},  $ \mathcal{C}_i^{\alpha_i'} $ is $ l $-Galois LCD code over $\mathbb{F}_q$. Take $ \alpha=\sum_{m=1}^4 \gamma_m \alpha_m' \in \mathcal{R}^n $ where $ \alpha_m'=\alpha_i' $ for $ m=i $ and $ \alpha_m'=(1,1,\dots,1) $ for $ m\neq i .$ By Theorem \ref{thm 4.4}, $ \mathcal{C}^{\alpha} $ is an $ [n,k,d] $ $ l $-Galois LCD code over the ring $ \mathcal{R} .$ 
	\end{proof}
	\section{MDS Code over $ \mathbf{\mathcal{R}} $}\label{mdscode}
	For a linear code $\mathcal{C}$  over the ring $\mathcal{R}$ with parameters $[n,k,d]$, we have  $|\mathcal{C}|\leq|\mathcal{R}|^{n-d+1}$  implies $d\leq n-\log_{|\mathcal{R}|}|\mathcal{C}|+1$, which is known as Singleton bound on ring $\mathcal{R}$.  Since $|\mathcal{R}|=q^4$ and $|\mathcal{C}|=q^k$, where $k=\sum_{i=1}^{4}k_i$, thus the Singleton bound is $d\leq n-\frac{1}{4}\sum_{i=1}^{4}k_i+1.$
	The code which attains the Singleton bound is called the MDS code. We have the following result on MDS code for finite field $\mathbb{F}_q$, see \cite{galoislcd}
	\begin{lemma}
		If $C$ is  a linear code over $\mathbb{F}_q$, then the following are equivalent
		\begin{enumerate}
			\item $C$ is MDS code over $\mathbb{F}_q.$
			\item $C^{\perp}$ is MDS code over $\mathbb{F}_q.$
			\item $C^{\perp_l}$ is MDS code  over $\mathbb{F}_q.$
		\end{enumerate}
	\end{lemma}
	
	The following theorem shows that a linear code is an MDS code if and only if its Euclidean dual ($ l $-Galois) is an MDS code over the ring $ \mathcal{R}$.	
	\begin{theorem}
		If $ \mathcal{C}= \bigoplus_{i=1}^4\gamma_i \mathcal{C}_i $ is a linear code over the ring $\mathcal{R}$, where  $ \mathcal{C}_i $'s are component codes over the finite field $ \mathbb{F}_q $ then
		\begin{enumerate}
			\item \label{mds1} $\mathcal{C}$ is an MDS code over the ring $\mathcal{R}$ if and only if $\mathcal{C}_i$'s are MDS codes over $\mathbb{F}_q$ with same parameters for each $i=1,2,3,4.$
			\item $\mathcal{C}$ is an MDS code over the ring $\mathcal{R}$ if and only if $\mathcal{C}^\perp$ is an MDS code over $\mathcal{R}$.
			\item  $\mathcal{C}$ is an MDS code over the ring $\mathcal{R}$ if and only if $\mathcal{C}^{\perp_l}$ is an MDS code over $\mathcal{R}$.
		\end{enumerate}
	\end{theorem}
	\begin{proof}
		(1)	Suppose $\mathcal{C}$ is MDS code over the  ring $\mathcal{R}$ with parameter $[n,k,d]$, where $4d=4n-\sum_{i=1}^{4}k_i+4$. Since $d=\min\limits_{1\leq i\leq 4}\{d_i\}$, where $ d_i=d_H(\mathcal{C}_i) $, it follows that $d=d_j$ for some $j=1,2,3,4$. This implies $4d_j=4n-\sum_{i=1}^{4}k_i+4$. Now we know that $d_i\leq n-k_i+1$ for $i=1,2,3,4$ and so $\sum_{i=1}^{4}d_i\leq 4n-\sum_{i=1}^{4}k_i+4=4d_j$. Since $d_j$ is minimum of $d_i$'s for $i=1,2,3,4$ therefore, $4d_j\leq\sum_{i=1}^{4}d_i$. Which conclude that $4d_j=\sum_{i=1}^{4}d_i$. It is only possible when $d_1=d_2=d_3=d_4.$ Hence $\mathcal{C}_i$'s are MDS codes over $\mathbb{F}_q$ with  the same parameters. \\ 
		Conversely, if $\mathcal{C}_i$'s are MDS codes with  the same parameters that is $d_1=d_2=d_3=d_4$ and $d_i=n-k_i+1\implies4d_i=4n-\sum_{i=1}^{4}k_i+4$ for $i=1,2,3,4$. Since $d=\min\limits_{1\leq i\leq 4}\{d_i\}$ this implies $4d=4n-\sum_{i=1}^{4}k_i+4$. Hence $\mathcal{C}$ is an MDS code. \\
		(2) Let $\mathcal{C}$ be an MDS code over the ring $\mathcal{R}$ then by  \ref{mds1}, $\mathcal{C}_i$'s are MDS codes over $\mathbb{F}_q$  having same parameters  for each $i=1,2,3,4$. Hence $\mathcal{C}_i^{\perp}$'s are  MDS codes over  $\mathbb{F}_q$ with same parameters for each $i=1,2,3,4$ this implies $\mathcal{C}^{\perp}$ is an MDS code over the ring $\mathcal{R}$. Similar argument can be made for converse.\\
		(3) Let $\mathcal{C}$ be an  MDS code over the ring $\mathcal{R}$, then $\mathcal{C}^{p^{(e-l)}}$ is also an MDS code over the ring $\mathcal{R}$, using  Lemma \ref{lemma3.1}, we get $\mathcal{C}^{\perp_l}=\left(\mathcal{C}^{p^{(e-l)}}\right) ^\perp$, hence  $\mathcal{C}^{\perp_l}$ is the  MDS code over $\mathcal{R}$. Conversely, assume $\mathcal{C}^{\perp_l}$ is MDS code over $\mathcal{R}$ it follows that $ \mathcal{C}^{p^{(e-l)}} $ MDS. Hence $\mathcal{\mathcal{C}}$ is MDS code over $\mathcal{R}$.
	\end{proof}
	\begin{remark}
		Result \ref{mds1} in the above theorem has been done over the ring $\mathbb{Z}_4+u\mathbb{Z}_4+v\mathbb{Z}_4+uv\mathbb{Z}_4$ in \cite{someresultsz4}. 	
	\end{remark}	
	\section{Conclusion}\label{con}
	We have studied $l$-Galois LCD codes over $ \mathcal{R}=\mathbb{F}_q+u\mathbb{F}_q+v\mathbb{F}_q+ uv\mathbb{F}_q.$ We have given a  correspondence relation between a linear code  and  its component codes to be  $l$-Galois LCD codes. We have also given  construction of $l$-Galois LCD code   from linear codes over $ \mathcal{R} .$ Consequently, we have shown that  for any linear code over $ \mathcal{R}, $ there exists  equivalent  Euclidean LCD (for $ q>3 $) and $l$-Galois LCD code (for $0<l<e$ and $p^{e-l}+1\mid p^e-1$ ) over $ \mathcal{R} .$ After that, We have given some results on MDS codes over $ \mathcal{R} .$  Further, the study of $l$-Galois LCD MDS code over $ \mathcal{R} $  could be an interesting problem. One can also investigate $l$-Galois LCD codes over different rings.
	\bibliographystyle{plain}
	\bibliography{Galois_LCD_codes_over_Fq+uFq+vFq+uvFq}

\begin{thebibliography}{10}

\bibitem{quantumcode}
Mohammad Ashraf and Ghulam Mohammad.
\newblock Quantum codes from cyclic codes over {$\Bbb F_q+u\Bbb F_q+v\Bbb
  F_q+uv\Bbb F_q$}.
\newblock {\em Quantum Inf. Process.}, 15(10):4089--4098, 2016.

\bibitem{sideattackcarlet}
Claude Carlet and Sylvain Guilley.
\newblock Complementary dual codes for counter-measures to side-channel
  attacks.
\newblock {\em Adv. Math. Commun.}, 10(1):131--150, 2016.

\bibitem{euclideanhermitian}
Claude Carlet, Sihem Mesnager, Chunming Tang, and Yanfeng Qi.
\newblock Euclidean and {H}ermitian {LCD} {MDS} codes.
\newblock {\em Des. Codes Cryptogr.}, 86(11):2605--2618, 2018.

\bibitem{newchpr}
Claude Carlet, Sihem Mesnager, Chunming Tang, and Yanfeng Qi.
\newblock New characterization and parametrization of {LCD} codes.
\newblock {\em IEEE Trans. Inform. Theory}, 65(1):39--49, 2019.

\bibitem{cafinite}
Claude Carlet, Sihem Mesnager, Chunming Tang, Yanfeng Qi, and Ruud Pellikaan.
\newblock Linear codes over {$\Bbb F_q$} are equivalent to {LCD} codes for
  {$q>3$}.
\newblock {\em IEEE Trans. Inform. Theory}, 64(4, part 2):3010--3017, 2018.

\bibitem{newmds}
Bocong Chen and Hongwei Liu.
\newblock New constructions of {MDS} codes with complementary duals.
\newblock {\em IEEE Trans. Inform. Theory}, 64(8):5776--5782, 2018.

\bibitem{kgaloisintroduce}
Yun Fan and Liang Zhang.
\newblock Galois self-dual constacyclic codes.
\newblock {\em Des. Codes Cryptogr.}, 84(3):473--492, 2017.

\bibitem{constructionmds}
Lingfei Jin.
\newblock Construction of {MDS} codes with complementary duals.
\newblock {\em IEEE Trans. Inform. Theory}, 63(5):2843--2847, 2017.

\bibitem{ConstacyclicCO}
Jo{\"e}l Kabor{\'e} and Mohammed~Elhassani Charkani.
\newblock Constacyclic codes over {$\Bbb F_q+u\Bbb F_q+v\Bbb F_q+uv\Bbb F_q$}.
\newblock {\em arXiv: Information Theory}, 2015.

\bibitem{someresultsz4}
Ping Li, Xuemei Guo, Shixin Zhu, and Xiaoshan Kai.
\newblock Some results on linear codes over the ring
  {$\Bbb{Z}_4+u\Bbb{Z}_4+v\Bbb{Z}_4+uv\Bbb{Z}_4$}.
\newblock {\em J. Appl. Math. Comput.}, 54(1-2):307--324, 2017.

\bibitem{galoislcd}
Xiusheng Liu, Yun Fan, and Hualu Liu.
\newblock Galois {LCD} codes over finite fields.
\newblock {\em Finite Fields Appl.}, 49:227--242, 2018.

\bibitem{lcdovercommring}
Zihui Liu and Jinliang Wang.
\newblock Linear complementary dual codes over rings.
\newblock {\em Des. Codes Cryptogr.}, 87(12):3077--3086, 2019.

\bibitem{massey}
James~L. Massey.
\newblock Linear codes with complementary duals.
\newblock volume 106/107, pages 337--342. 1992.
\newblock A collection of contributions in honour of Jack van Lint.

\bibitem{anote}
Rongsheng Wu and Minjia Shi.
\newblock A note on {$k$}-{G}alois {LCD} codes over the ring {$\Bbb {F}_q+u\Bbb
  F_q$}.
\newblock {\em Bull. Aust. Math. Soc.}, 104(1):154--161, 2021.

\bibitem{conditionforcycliccodemassey}
Xiang Yang and James~L. Massey.
\newblock The condition for a cyclic code to have a complementary dual.
\newblock {\em Discrete Math.}, 126(1-3):391--393, 1994.

\bibitem{skewc}
Ting Yao, Minjia Shi, and Patrick Sol\'{e}.
\newblock Skew cyclic codes over {$\Bbb F_q+u\Bbb F_q+v\Bbb F_q+uv\Bbb F_q$}.
\newblock {\em J. Algebra Comb. Discrete Struct. Appl.}, 2(3):163--168, 2015.

\end{thebibliography}
   Astha Agrawal \\
   Department of Mathematics \\
   Indian Institute of Technology Delhi,\\
   New Delhi, 110016, India\\
   \textit{Email Address:} \url{asthaagrawaliitd@gmail.com}\\\\
   Gyanendra K. Verma \\
   Department of Mathematics \\
   Indian Institute of Technology Delhi,\\
   New Delhi, 110016, India\\
  \textit{ Email Address:} \url{gkvermaiitdmaths@gmail.com}\\\\
   R. K. Sharma \\
   Department of Mathematics \\
   Indian Institute of Technology Delhi,\\
   New Delhi, 110016, India\\
   \textit{Email Address:} \url{rksharmaiitd@gmail.com}\\\\
   	
\end{document}